\documentclass[11pt]{article}
\usepackage{graphicx}
\usepackage{latexsym}
\usepackage{amssymb}
\usepackage{amsmath}
\usepackage{amsthm}
\usepackage{thmtools}
\makeatletter
\@ifundefined{newcounteralias}{}{%
  \renewcommand\thmt@autorefsetup{\@xa\def\csname\thmt@envname autorefname\@xa\endcsname\@xa{\thmt@thmname}}%
}
\makeatother
\usepackage{forloop}
\usepackage[paper=letterpaper,margin=1in]{geometry}
\usepackage[ruled,vlined,linesnumbered]{algorithm2e}
\usepackage{paralist}
\usepackage{xcolor}
\usepackage{hyperref}

\definecolor{darkgreen}{rgb}{0,0.5,0}
\hypersetup{
    unicode=false,            
    colorlinks=true,          
    linkcolor=blue!70!black,  
    citecolor=darkgreen,      
    filecolor=magenta,        
    urlcolor=blue!70!black    
}

\newtheorem{theorem}{Theorem}[section]
\newtheorem{lemma}[theorem]{Lemma}\RenewCommandCopy{\theHlemma}{\theHtheorem}
\RenewCommandCopy{\theHcorollary}{\theHtheorem}
\newtheorem{observation}[theorem]{Observation}\RenewCommandCopy{\theHobservation}{\theHobservation}

\newcommand{\defcal}[1]{\expandafter\newcommand\csname c#1\endcsname{{\mathcal{#1}}}}
\newcommand{\defbb}[1]{\expandafter\newcommand\csname b#1\endcsname{{\mathbb{#1}}}}
\newcommand{\defvec}[1]{\expandafter\newcommand\csname v#1\endcsname{{\mathbf{#1}}}}
\newcounter{calBbCounter}
\forLoop{1}{26}{calBbCounter}{
    \edef\letter{\alph{calBbCounter}}
		\edef\Letter{\Alph{calBbCounter}}
    \expandafter\defcal\Letter
		\expandafter\defbb\Letter
		\expandafter\defvec\letter
}

\newcommand{\nnR}{{\bR_{\geq 0}}}

\newcommand{\characteristic}{{\mathbf{1}}}
\newcommand{\email}[1]{{\href{mailto:#1}{#1}}}
\newcommand{\groundset}{N}

\usepackage[color=lightgray]{todonotes}

\begin{document}

\title{On Maximizing a Weakly Submodular Function over\\a Matroid Constraint via the Greedy Algorithm}
\author{Moran Feldman\footnote{Department of Computer Science, University of Haifa. This work was done while the author was visiting Queen Mary University of London. E-mail: \email{moranfe@cs.haifa.ac.il}} \and Justin Ward\footnote{School of Mathematics, Queen Mary University of London. E-mail: \email{justin.ward@qmul.ac.uk}}}

\maketitle
\thispagestyle{empty}
\pagenumbering{Alph}
\begin{abstract}

We consider the problem of approximately maximizing a weakly submodular function using the standard greedy algorithm, which is known to give tight approximation results for such functions under a cardinality constraint. We show that this is not the case for general matroid constraints. For any $\gamma < 1$, we give a family of $\gamma$-weakly submodular functions and a simple partition matroid constraint and show that the standard greedy algorithm provides no constant approximation for the resulting constrained maximization problem.

\medskip

\noindent \textbf{keywords:} submodularity ratio, weakly submodular function, matroid constraint, greedy algorithm

\end{abstract}

\newpage
\pagenumbering{arabic}

\section{Introduction}

Submodular functions are an important family of set functions\footnote{A set function is a function $f\colon 2^\groundset \to \bR$ assigning a numerical value to every subset of some ground set $\groundset$.} formalizing the notion of diminishing returns. Naturally, such functions are prevalent in economy and game theory~\cite{nisan2007algorithmic}, but also encompass prominent functions from combinatorial optimization such as cut functions, coverage functions and matroid rank functions~\cite{buchbinder2018submodular}. Furthermore, because the amount of information provided by a set of objects often exhibits diminishing returns, optimization of submodular functions has found applications also in other fields such as machine learning, image processing and decision and control~\cite{bach2013learning,bilmes2022submodularity,dutta2026submodular,krause2011submodularity}.

For many submodular maximization problems the canonical algorithm is the natural greedy algorithm. The greedy algorithm iteratively adds elements to its solution, picking in each iteration the element that increases the value of the solution by the most (to ensure feasibility, the greedy algorithm considers for selection in each iteration only elements that can be added without violating the constraint). It is well-known that the greedy algorithm obtains an approximation ratio of $1 - 1/e$ for maximization of a monotone submodular function over a cardinality constraint, which is the best ratio obtainable by a polynomial time algorithm~\cite{nemhauser1978best,nemhauser1978analysis}. The greedy algorithm is useful also for more involved kinds of constraints, although its guarantee often degrades. For example, in $1978$, it was shown that the greedy algorithm guarantees $1/2$-approximation for maximizing a monotone submodular function over a general matroid constraint~\cite{fisher1978analysis}. This guarantee was the state-of-the-art until the introduction of the randomized Continuous Greedy algorithm roughly $30$ years later, which attains the optimal $(1 - 1/e)$-approximation~\cite{calinescu2011maximizing}. Later works also managed to improve over the greedy algorithm using deterministic algorithms~\cite{buchbinder2024deterministic,buchbinder2023deterministic}.

In practice, the greedy algorithm is often very successful even when the objective function is not submodular. Das and Kempe~\cite{das2018approximate} introduced an important theoretical tool for understanding this phenomenon. Specifically, they defined the \emph{submodularity ratio}, which is a number between $0$ and $1$ measuring the degree of submodularity of a monotone set function $f$. This ratio compares the increase in $f$ that can be realized by adding a set of elements to some base set to the sum of the individual increases realized when each element is added to the same set individually (see \eqref{eq:gamma-ws} for a formal definition). The resulting ratio is 1 for submodular functions, and it becomes smaller as the function exhibits more and more non-submodular behavior (i.e., as the gain realized by adding some set of elements further and further exceeds the sum of these individual elements' marginal values). Das and Kempe~\cite{das2018approximate} proved that the greedy algorithm guarantees $(1 - e^{-\gamma})$-approximation for maximizing monotone set functions with a submodularity ratio $\gamma$ over a cardinality constraint (such functions are also called \emph{$\gamma$-weakly submodular functions}).
Later, Harshaw et al.~\cite{harshaw2019submodular} proved that this guarantee of the greedy algorithm is the best that can be achieved by any polynomial time algorithm.

The original motivation of Das and Kempe for introducing the submodularity ratio was sparse approximation over dictionary vectors, but later works showed that this ratio is useful also in many other applications (a few examples of these include~\cite{agrawal2026minibatch,bian2017guarantees,elenberg2017streaming,harshaw2019submodular,khanna2017scalable,santiago2020weakly,tofigh2026acli,zhu2024regularized}). Despite its richness, almost all the work on the submodularity ratio considers only cardinality constraints, and it is unknown whether the greedy algorithm has any non-trivial guarantees for maximizing $\gamma$-weakly submodular functions over other types of constraints. Chen et al.~\cite{chen2108weakly} showed that a randomized version of the greedy algorithm guarantees $(1 + 1/\gamma)^{-2}$-approximation for maximization of such functions over a matroid constraint. Notice that this approximation ratio is constant in the sense that is independent of the size of the instance. Chen et al.\ explicitly leave open the question of whether the standard greedy algorithm may already attain some, potentially weaker but still constant factor, approximation in the matroid case.

More restrictive notions of approximate submodularity have also been proposed. The \emph{generic submodular ratio}~\cite{gong2021maximize,gong2019parametric} and \emph{inverse curvature}~\cite{bogunovic2018robust} both bound the maximum factor by which \emph{any individual element's} marginal value may increase as other elements are added to the solution. Gong et al.~\cite{gong2021maximize} show that the greedy algorithm gives a $\frac{\bar{\gamma}}{1 + \bar{\gamma}}$-approximation for maximizing a function with generic submodular ratio $\bar{\gamma}$ under a single matroid constraint, as well as a constant-factor approximation (depending on $\bar{\gamma}$) under knapsack and multiple matroid constraints. Kuhnle et al.\ proposed \emph{weak DR-submodularity}, which imposes a similar element-wise constraint on continuous functions and functions on the integer lattice~\cite{kuhnle2018fast}. It is easy to show that all of these element-wise conditions imply bounds on the submodularity ratio $\gamma$. However, the converse is not true, even for several regression problems originally motivating the submodularity ratio $\gamma$~\cite{thiery2022two}. Ward and Thiery~\cite{thiery2022two} introduced a complementary parameter $\beta$ to the submodularity ratio that considers the aggregate effect of element \emph{removals}. They showed that this quantity $\beta$ can be bounded analogously to the submodularity ratio in many of the problems considered in~\cite{das2018approximate} and gave algorithms for matroid-constrained problems that improve on the approximation performance of Chen et al.~\cite{chen2108weakly} when this is the case. Finally, we note that more \emph{general} notions of approximate submodularity have also been introduced, which can be used to obtain guarantee for \emph{non-monotone} functions in the cardinality constrained case~\cite{santiago2020weakly}.

Despite the above work, it is still unknown whether the standard greedy algorithm may in fact attain a constant-factor approximation for the original, general class of $\gamma$-weakly submodular functions proposed by Das and Kempe~\cite{das2018approximate} without further restrictions. In this paper, we provide a negative answer for this question, showing that for any $\gamma \in (0, 1)$, the approximation ratio of the greedy algorithm for matroid constraints must deteriorate with the size of the ground set. This implies that, unlike the case of cardinality constraints, either additional algorithmic insights or further restrictions are indeed necessary to obtain \emph{any constant factor approximation} for $\gamma$-weakly submodular functions under matroid constraints. More formally, we prove the following theorem.

\begin{theorem} \label{thm:greedy_bad_example}
For evry value $\gamma \in (0, 1)$ and positive integer $n$, there exists a non-negative monotone $\gamma$-weakly submodular function $f \colon 2^\groundset \to \nnR$ and a simplified partition matroid $\cM = (\groundset, \cI)$ of rank $n-1$ over a ground set $\groundset$ of size $n$ such that, when using the greedy algorithm to maximize $f$ subject to $\cM$, the approximation ratio of the solution obtained is $O(\frac{\gamma}{(1 - \gamma)^2 \log n})$.
\end{theorem}


\section{Preliminaries}

\paragraph{Set Functions.}

Given a set function $f\colon 2^\groundset \to \bR$, a set $S \subseteq \groundset$ and an element $u \in \groundset$, we denote $f(u \mid S) \triangleq f(S \cup \{u\}) - f(S)$, and we refer to this value as the \emph{marginal contribution} of $u$ to the set $S$ with respect to $f$. Similarly, given an additional set $T \subseteq \groundset$, we denote by $f(T \mid S) \triangleq f(S \cup T) - f(S)$ the marginal contribution of $T$ to $S$ with respect to $f$.

The set function $f$ is monotone if $f(u \mid S) \geq 0$ for every $S \subseteq \groundset$ and $u \in \groundset \setminus S$. The submodularity ratio of a monotone $f$ is the largest value $\gamma \in [0, 1]$ for which
\begin{equation}
	\gamma \cdot f(T \mid S)
	\leq
	\sum_{u \in T \setminus S} f(u \mid S)
	\qquad
	\forall\; S \subseteq T \subseteq \groundset
	\enspace.\label{eq:gamma-ws}
      \end{equation}
As mentioned above, we say that $f$ is $\gamma$-weakly submodular if its submodularity ratio is at least $\gamma$. The function $f$ is submodular if $f(u \mid S) \geq f(u \mid T)$ for every two sets $S \subseteq T \subseteq \groundset$ and element $u \in \groundset \setminus T$. It is known that a monotone function is $1$-weakly submodular if and only if it is submodular~\cite{das2018approximate}.

\paragraph{Matroids.}

A matroid is defined by a pair $(\groundset, \cI)$, where $\groundset$ is the ground set of the matroid and $\cI \subseteq 2^\groundset$ is a collection of subsets of $\groundset$ obeying the following three properties:
{
\setlength{\pltopsep}{5pt}
\setlength{\plitemsep}{0pt}
\begin{compactitem}
	\item The set $\varnothing$ belongs to $\cI$.
	\item For every two sets $S \subseteq T \subseteq \groundset$, if $T \in \cI$, then $S \in \cI$.
	\item For every two sets $S, T \subseteq \groundset$, if $|S| < |T|$, then there exists an element $u \in T \setminus S$ such that $S \cup \{u\} \in \cI$.
\end{compactitem}
}
The sets in the collection $\cI$ are the sets that are feasible under the constraint defined by this matroid. Following linear algebra motivations for the study of matroids, it is customary to refer to these sets as the \emph{independent sets} of the matroid.

In this paper, we only consider a particular family of matroids known as \emph{simplified partition matroids}. A simplified partition matroid of rank $k$ is defined using a partition of the ground set $\groundset$ into $k$ disjoint parts $\groundset_1, \groundset_2, \dotsc, \groundset_k$. A set $S \subseteq \groundset$ is then independent in the simplified partition matroid if $|S \cap \groundset_i| = 1$ for every $i \in [k]$.


\paragraph{The Greedy Algorithm.}

A formal presentation of the greedy algorithm for maximizing a monotone submodular function over a matroid constraint can be found as Algorithm~\ref{alg:greedy}.

\begin{algorithm}
\caption{\textsc{Greedy for Matroid Constraints} $(\groundset, \cI, f)$} \label{alg:greedy}
Let $i \gets 0$ and $S_0 \gets \varnothing$.\\
\While{there exists an element $u \in \groundset \setminus S_i$ such that $S_i \cup \{u\} \in \cI$}
{
	Update $i \gets i + 1$.\\
	Let $u_i \in \arg\max_{u \in \groundset \mid S_{i - 1} \cup \{u\} \in \cI} f(u \mid S_{i - 1})$.\\
	Let $S_i \gets S_{i - 1} \cup u_i$.
}
\Return $S_i$.
\end{algorithm}


\section{Bad Instance}
\label{sec:hard-instance}

In this section, we prove Theorem~\ref{thm:greedy_bad_example} by constructing a instance that is bad for the greedy algorithm and demonstrates that the approximation guarantee of this algorithm is as poor as is stated in the theorem. The bad instance we construct is defined over a ground set $\groundset = Z \cup A$, where $Z = \{z_1,\dotsc,z_{n-2}\}$ and $A = \{o, v\}$. The matroid of the bad instance is a simplified partition matroid $\cM = (\groundset,\cI)$ defined by the partition $\cN_1, \cN_2, \dotsc, \cN_{n - 1}$, where $\groundset_i = \{z_i\}$ for $i \in [n-2]$ and $\groundset_{n-1} = A$. Notice that this implies that a set $S$ is independent in $\cM$ if and only if $|S \cap A| = 1$, and thus, $\cM$ has exactly $2$ bases (inclusion-wise maximal sets): $Z \cup \{v\}$ and $Z \cup \{o\}$. In the following paragraph, we define an objective function $f$ for our bad instance. This function ensures that the greedy algorithm constructs the first of the above bases, despite the second base being significantly more valuable.

As promised, we now define the non-negative monotone $\gamma$-weakly submodular function $f\colon 2^\groundset \to \nnR$ that serves as the objective function of our bad instance. Let us define first a function $g \colon \nnR \to \nnR$ and a sequence $(a_i)_{i \geq 1}$ of values by
\begin{align*}
g(x) &\triangleq \ln(1 + x) \enspace, & 
a_i &\triangleq \gamma \cdot \min\{g(i) - g(i-1), 2 i^{-\beta}\} \text{ for all } i \geq 1\enspace,
\end{align*}
where $\beta \triangleq \gamma + \gamma^{-1} - 1$. Note that $\beta > 1$ for any $\gamma \in (0, 1)$. The objective function $f\colon 2^\groundset \to \nnR$ is now given by
\begin{equation}
f(S) \triangleq 2a_1 \cdot |S \cap \{v\}| + a_1 \cdot |S \cap \{o\}| + \begin{cases}
g(|S \cap Z|)& \text{if $o \in S$}\enspace, \\
\sum_{i = 1}^{|S \cap Z|}a_i & \text{if $o \not\in S$}\enspace,
\end{cases}
\label{eq:f-definition}
\end{equation}
for every $S \subseteq \groundset$. Notice that for sets $S$ that include the element $o$, the objective function $f$ assigns a value of roughly $g(|S|)$, i.e., the sum of the discrete derivative $g(i) - g(i - 1)$ over all $i \in [|S|]$. For sets that do not include $o$, $f$ assigns roughly the value $\sum_{i = 1}^{|S|} a_i$. By definition, every $a_i$ is smaller than the corresponding discrete derivative $g(i) - g(i - 1)$, and thus, sets that include $o$ naturally tend to be more valuable compared to sets of the same size that do not include $o$. The following observation states that more formally.
\begin{observation} \label{obs:discrete_derivative}
For every integer $k \geq 0$, $\gamma \cdot g(k) = \gamma \cdot \sum_{i = 1}^k [g(i) - g(i - 1)] \geq \sum_{i = 1}^k a_i$.
\end{observation}
The gap shown by the last observation between $g(|S|)$ and $\sum_{i = 1}^{|S|} a_i$ is only a factor of $\gamma$. To prove that the instance we construct is bad, we need this multiplicative gap to increase with $|S|$. The definition of the series $\{a_i\}_{i \geq 1}$ was carefully chosen to allow for that, and at the same time, preserve the properties of $f$ promised by Theorem~\ref{thm:greedy_bad_example}, i.e., non-negativity, monotonicity and $\gamma$-weak submodularity. The following observation formally shows that $f$ has the first two of these properties.

\begin{observation}
The function $f$ defined by \eqref{eq:f-definition} is non-negative and monotone, and the values in the sequence $(a_i)_{i \geq 1}$ are all positive.
\label{obs:f-simple-properties}
\end{observation}
\begin{proof}
To see that the values in the sequence $(a_i)_{i \geq 1}$ are all positive, we note that for every $i \geq 1$, both $2i^{-\beta} > 0$ and $g(i) - g(i - 1) > 0$ (the second inequality holds since $g$ is an increasing function). This immediately proves that every term in the definition of $f(S)$ can only increase when an element $u \neq o$ is added to $S$. Thus, to show that $f$ is monotone, it only remains to observe that for any set $S \subseteq \groundset \setminus \{o\}$,
\[
	f(o \mid S)
	=
	a_1 + g(|S \cap Z|) - \sum_{i = 1}^{|S \cap Z|} a_i\\
	\geq
	a_1 + (1 - \gamma) \cdot g(|S \cap Z|)
	>
	0
	\enspace,
\]
where the first inequality follows from Observation~\ref{obs:discrete_derivative}.
The non-negativity of $f$ now follows from the monotonicity of $f$ since $f(\varnothing) = 0$ by $f$'s definition.
\end{proof}

Lemma~\ref{lem:f-weakly-submodular} below shows that the function $f$ is also $\gamma$-weakly submodular. However, before diving into the quite technical proof of this lemma, let us first show that Theorem~\ref{thm:greedy_bad_example} indeed follows if $f$ has this property.
\begin{proof}[Proof of Theorem~\ref{thm:greedy_bad_example}]
Consider the application of the greedy algorithm (Algorithm~\ref{alg:greedy}) to the above described simplified partition matroid $\cM = (\groundset,\cI)$  and the non-negative monotone $\gamma$-weakly submodular function $f\colon 2^\groundset \to \nnR$ defined by \eqref{eq:f-definition}. At the first iteration of the greedy algorithm, it starts with $S_0 = \varnothing$ and selects an element $u \in \cN$ maximizing $f(u \mid S_0)$. Since $f(v \mid S_0) = f(v \mid \varnothing) = 2a_1$ and $f(u \mid S_0) = f(u \mid \varnothing) = a_1$ for all elements $u \in \cN \setminus \{v\}$, the greedy algorithm selects $v$ in the first step. After that step, $o$ cannot be selected since only a single element of $A = \{v, o\}$ can appear in a feasible solution. Instead, the $n - 1$ elements of $Z$ are selected one by one, and thus, the output set of the greedy algorithm is $G = Z \cup \{v\}$. Note now that
\begin{multline*}
f(G) = 2a_1 + \sum_{i = 1}^{n-2}a_i \leq 3a_1 + \gamma \cdot \int_{1}^{n-2} 2 x^{-\beta}\,\mathrm{d}x
= 3a_1 + \gamma \cdot \left.\frac{2 x^{1-\beta}}{1-\beta}\right|_{x = 1}^{n-2} \\
\leq 3a_1 + \frac{2\gamma}{\beta-1} = 3\gamma \ln 2 + \frac{2\gamma^2}{(1 - \gamma)^2} \leq (3\ln 2 + 2)\cdot \frac{\gamma}{(1-\gamma)^2}\enspace.
\end{multline*}
On the other hand, consider the solution $O = Z \cup \{o\}$. For this solution,
\begin{equation*}
f(O) = a_1 + g(n-2) > \ln(n-1) \geq \frac{(1-\gamma)^2}{\gamma} \cdot \frac{\ln(n - 1)}{3 + 2\ln 2}\cdot f(G) = \Omega\bigg(\frac{(1 - \gamma)^2 \cdot \log n}{\gamma}\bigg) \cdot f(G)\enspace.\qedhere
\end{equation*}
\end{proof}

As mentioned above, Lemma~\ref{lem:f-weakly-submodular} below is devoted to proving that $f$ is $\gamma$-weakly submodular. The proof of this lemma uses the properties of the series $\{a_i\}_{i \geq 1}$ given by the next observation.
\begin{observation} \label{obs:a-term-bound}
The sequence $(a_i)_{i \geq 1}$ is decreasing, and for every integer $i \geq 1$, $a_i \leq \gamma / i$.
\end{observation}
\begin{proof}
Recall that $a_i = \gamma \cdot \min\{g(i) - g(i-1), 2 i^{-\beta}\}$. The term $2 i^{-\beta}$ is decreasing in $i$ since $\beta > 1$, and the term $g(i) - g(i - 1)$ is also decreasing since $g$ is concave. Thus, $(a_i)_{i \geq 1}$ is indeed a decreasing series. Additionally,
\[
	a_i \leq \gamma \cdot [g(i) - g(i-1)] = \gamma \cdot [\ln(1 + i) - \ln(i)] = \gamma \cdot \ln\left(1 + \frac{1}{i}\right) \leq \frac{\gamma}{i}\enspace.
	\qedhere
\]
\end{proof}

\begin{lemma}
The function $f$ defined by \eqref{eq:f-definition} is $\gamma$-weakly submodular.
\label{lem:f-weakly-submodular}
\end{lemma}
\begin{proof}
To prove the lemma, we need to show that \eqref{eq:gamma-ws} holds for every $S \subseteq T \subseteq \groundset$. We consider $3$ cases based on whether $o$ is present in $S$ or $T$. In the first two cases, we show that $f$ in fact behaves as a submodular function. The third case is significantly more involved.

\paragraph{Case 1: \boldmath$o \not\in T$.} In this case,
\begin{align*}
\gamma \cdot f(T \mid S) \leq f(T \mid S) &= 2a_1 \cdot \characteristic[v \in T \setminus S] + \sum_{i = |S\cap Z| + 1}^{|T \cap Z|} \mspace{-18mu} a_i \\
&\leq 
2a_1 \cdot \characteristic[v \in T \setminus S] + |(T \setminus S) \cap Z| \cdot a_{|S \cap Z|+1} = \sum_{u \in T \setminus S}f(u \mid S)\enspace,
\end{align*}
where the inequality follows from the fact that $(a_i)_{i \geq 1}$ are non-increasing (Observation~\ref{obs:a-term-bound}).

\paragraph{Case 2: \boldmath$o \in S$.} Similarly to the previous case,
\begin{align*}
\gamma \cdot f(T \mid S) &\leq f(T \mid S) = 2a_1 \cdot \characteristic[v \in T \setminus S] + g(|T \cap Z|) - g(|S \cap Z|) \\
&\leq 
2a_1 \cdot \characteristic[v \in T \setminus S] + |(T \setminus S) \cap Z| \cdot \Bigl(g(|S \cap Z| + 1) - g(|S \cap Z|)\Bigr) = \sum_{u \in T \setminus S}f(u \mid S)\enspace,
\end{align*}
where the inequality follows from the fact that $g$ is concave.

\paragraph{Case 3: \boldmath$o \in T \setminus S$.} In this case, we first observe that 
\[
f(T \mid S) = a_1 \cdot (1 + 2 \cdot \characteristic[v \in T \setminus S]) + g(|T \cap Z|) - \sum_{i = 1}^{|S \cap Z|}a_{i}\enspace,
\]
and for $u \not \in S$,
\[
f(u \mid S) = 
\begin{cases}
  a_{1} + g(|S \cap Z|) - \sum_{i = 1}^{|S \cap Z|} a_{i}\enspace, &u = o \\
  2a_{1}\enspace,& u = v \\
  a_{|S \cap Z| + 1}\enspace, &u \in Z\enspace.
\end{cases}
\]
Therefore, $\gamma \cdot f(T \mid S) \leq \sum_{u \in T \setminus S}f(u \mid S)$ if and only if
\begin{multline*}
\gamma \cdot a_{1}\cdot (1 + 2 \cdot \characteristic[v \in T \setminus S]) + \gamma \cdot g(|T \cap Z|) - \gamma \sum_{i = 1}^{|S \cap Z|}a_{i} \\ \leq 
\bigg(a_{1} + g(|S \cap Z|) - \sum_{i = 1}^{|S \cap Z |}a_{i}\bigg) + 2a_1 \cdot \characteristic[v \in T \setminus S]) + |(T \setminus S) \cap Z|\cdot a_{|S \cap Z| + 1} \enspace.
\end{multline*}
After rearranging this inequality and noting that $\gamma \cdot a_{1} < a_{1}$, we get that it suffices to show that
\begin{equation}
\gamma \cdot g(|T \cap Z|) - g(|S \cap Z|) + (1 - \gamma) \sum_{i = 1}^{|S \cap Z|}a_{i} \leq |(T \setminus S) \cap Z|\cdot a_{|S \cap Z| + 1}\enspace.
\label{case-3-main-goal}
\end{equation}
If $(T \setminus S) \cap Z = \varnothing$, Inequality~\eqref{case-3-main-goal} is equivalent to 
$\sum_{i = 1}^{|S\cap Z|}a_i \leq g(|S \cap Z|)$, which follows from Observation~\ref{obs:discrete_derivative}. Otherwise, let us set $x \triangleq |S \cap Z| \geq 0$, $y \triangleq |(T \setminus S) \cap Z| \geq 1$, and $h(x) \triangleq (1-\gamma)\sum_{i = 1}^xa_i - g(x)$. This allows us to re-express~\eqref{case-3-main-goal} as
\begin{equation}
\frac{\gamma \cdot g(x+y) + h(x)}{y} \leq a_{x + 1}\enspace.\label{eq:x-y-bound-inequality}
\end{equation}
The remainder of the proof is devoted to showing that~\eqref{eq:x-y-bound-inequality} holds for any $x \geq 0$ and $y\geq 1$. 
First, we observe that Observation~\ref{obs:discrete_derivative} implies that
\begin{equation}
h(x) = (1-\gamma)\sum_{i = 1}^x a_i - g(x)  \leq \bigl[(1-\gamma)\cdot \gamma  - 1  \bigr]\cdot g(x) = (\gamma -\gamma^2 - 1) \cdot g(x) =  -\gamma\beta \cdot g(x) \enspace.
\label{eq:h-bound}
\end{equation}
If $a_{x + 1} = \gamma \cdot [g(x+1) - g(x)]$, then~\eqref{eq:x-y-bound-inequality} holds since
\begin{equation*}
\frac{\gamma \cdot g(x+y) + h(x)}{y} \leq 
\frac{\gamma \cdot g(x+y) - \gamma \beta \cdot g(x)}{y} \leq 
\frac{\gamma \cdot [g(x+y) - g(x)]}{y} \leq 
\gamma \cdot[g(x+1) - g(x)]\enspace,
\end{equation*}
where the first inequality follows from~\eqref{eq:h-bound}, the second from $\beta > 1$ and the non-negativity of $g$, and the last from the concavity of $g$.

In the remaining case, we have $a_{x+1} = 2\gamma\cdot (x+1)^{-\beta}$.
We note that the value of $h(x)$ does not depend on $y$, and hence, taking the derivative of the left hand side~\eqref{eq:x-y-bound-inequality} with respect to $y$ yields
\begin{equation}
\frac{\gamma \cdot y \cdot g'(x + y) - \gamma \cdot 
g(x+y) - h(x)}{y^2}\enspace.
\label{eq:lhs-derivative}
\end{equation}
Since $g(x+y) = \ln(1 + x + y)$, we
have $\gamma \cdot y \cdot g'(x + y) = \frac{\gamma \cdot y}{1 + x + y} \in (0, \gamma)$ for all $y \geq 1$. Thus, whenever $y < e^{-\frac{h(x)}{\gamma}}- 1 - x$ the numerator of~\eqref{eq:lhs-derivative} is larger than
\begin{equation*}
0 - \gamma \cdot g(x+y) - h(x) 
= - \gamma \cdot \ln(1 + x + y) - h(x) 
> - \gamma \cdot \ln\Bigl(e^{-\frac{h(x)}{\gamma}}\Bigr) - h(x) = 0\enspace,
\end{equation*}
and hence, the left hand side of~\eqref{eq:x-y-bound-inequality} is increasing in $y$.
Similarly, whenever $y > e^{1-\frac{h(x)}{\gamma}}- 1 - x$, the numerator of~\eqref{eq:lhs-derivative} is smaller than
\begin{equation*}
\gamma - \gamma \cdot g(x+y) - h(x) 
= \gamma - \gamma \cdot \ln(1 + x + y) - h(x) 
< \gamma - \gamma \cdot \ln\Bigl(e^{1 - \frac{h(x)}{\gamma}}\Bigr) - h(x) = 0 \enspace,
\end{equation*}
and hence, the left hand side of~\eqref{eq:x-y-bound-inequality} is decreasing in $y$. It follows that the left side of~\eqref{eq:x-y-bound-inequality} achieves its maximum value at some $y^*$ with $e^{-\frac{h(x)}{\gamma}} - 1 - x \leq y^* \leq e^{1-\frac{h(x)}{\gamma}} - 1 - x$. Using the upper bound in the numerator of~\eqref{eq:x-y-bound-inequality} and the lower bound in the denominator, this maximum value is at most
\begin{equation*}
\frac{\gamma \cdot g(x + e^{1-\frac{h(x)}{\gamma}} - 1 - x) + h(x)}{e^{-\frac{h(x)}{\gamma}} - 1 - x}
= 
\frac{\gamma \cdot \ln\Bigl(e^{1-\frac{h(x)}{\gamma}}\Bigr) + h(x)}{e^{-\frac{h(x)}{\gamma}} - 1 - x} 
=
\frac{\gamma}{e^{\frac{-h(x)}{\gamma}} - 1 - x}\enspace.
\end{equation*}
Now, we consider the denominator of this last expression. By~\eqref{eq:h-bound}, we have 
\begin{equation*}
e^{\frac{-h(x)}{\gamma}} \geq e^{\beta \cdot g(x)} = e^{\beta\cdot \ln(1+x)} = (1+x)^\beta = \frac{2\gamma}{a_{x+1}}\enspace,
\end{equation*}
and by Observation~\eqref{obs:a-term-bound}, we have $x + 1 \leq \frac{\gamma}{a_{x + 1}}$. Thus, altogether, we have
\begin{equation*}
\frac{\gamma \cdot g(x + y) + h(x)}{y} \leq \frac{\gamma}{e^{\frac{-h(x)}{\gamma}} - 1 - x} \leq \frac{\gamma}{\frac{2\gamma}{a_{x+1}} - \frac{\gamma}{a_{x+1}}} = a_{x+1}\enspace,\end{equation*}
which completes the proof of Inequality~\eqref{eq:x-y-bound-inequality}.
\end{proof}


\bibliographystyle{plainurl}
\bibliography{WeakSubmodularity}

\end{document}